\documentclass[12pt]{elsarticle} 
\usepackage[margin=2.5cm]{geometry}% by courtesy of Mico
\usepackage{lipsum}

\makeatletter
\def\ps@pprintTitle{%
   \let\@oddhead\@empty
   \let\@evenhead\@empty
   \def\@oddfoot{\reset@font\hfil\thepage\hfil}
   \let\@evenfoot\@oddfoot
}
\makeatother
\usepackage{amsfonts,color,morefloats,pslatex}
\usepackage{amssymb,amsthm, amsmath,latexsym}

\newtheorem{theorem}{Theorem}
\newtheorem{lemma}[theorem]{Lemma}

\newtheorem{open}[theorem]{Open Problem}

\newtheorem{example}[theorem]{Example}

\newcommand\myatop[2]{\genfrac{}{}{0pt}{}{#1}{#2}}

\newcommand{\lcm}{{\mathrm{lcm}}}
\newcommand{\tr}{{\mathrm{Tr}}}

\newcommand{\gf}{{\mathrm{GF}}}

\newcommand{\Aut}{{\mathrm{Aut}}} 
\newcommand{\PAut}{{\mathrm{PAut}}} 
\newcommand{\MAut}{{\mathrm{MAut}}} 
\newcommand{\GAut}{{\mathrm{Aut}}}

\newcommand{\Sym}{{\mathrm{Sym}}}

\newcommand{\wt}{{\mathtt{wt}}}

\newcommand{\Z}{\mathbb{{Z}}}

\newcommand{\m}{\mathbb{M}}

\newcommand{\cP}{{\mathcal{P}}} 
\newcommand{\cB}{{\mathcal{B}}}

\newcommand{\C}{{\mathsf{C}}}

\newcommand{\cR}{{\mathcal{R}}}

\newcommand{\bc}{{\mathbf{c}}} 
\newcommand{\bg}{{\mathbf{g}}}

\newcommand{\bone}{{\mathbf{\bar{1}}}}

\newcommand{\bD}{{\mathbb{D}}}

\newcommand{\GA}{{\mathrm{GA}}}

\usepackage{blindtext}

\begin{document}
%\tableofcontents

\begin{frontmatter}



%\title{The linear codes of $2$-designs held in a class of ternary linear codes    

\title{Linear codes of $2$-designs associated with subcodes
of the ternary generalized Reed-Muller codes}

\tnotetext[fn1]{C. Ding's research was supported by the Hong Kong Research Grants Council,
Proj. No. 16300418. C. Tang was supported by National Natural Science Foundation of China (Grant No.
11871058) and China West Normal University (14E013, CXTD2014-4 and the Meritocracy Research
Funds).}
%}

\author[cding]{Cunsheng Ding}
\ead{cding@ust.hk}
\author[cmt]{Chunming Tang}
\ead{tangchunmingmath@163.com}
\author[vdt]{Vladimir D. Tonchev}
\ead{tonchev@mtu.edu}

%\cortext[lcj]{Corresponding author}
\address[cding]{Department of Computer Science and Engineering, The Hong Kong University of Science and Technology, Clear Water Bay, Kowloon, Hong Kong, China}
\address[cmt]{School of Mathematics and Information, China West Normal University, Nanchong, Sichuan,  637002, China}
\address[vdt]{
Department of Mathematical Sciences, Michigan Technological University, 
Houghton, Michigan 49931, USA
}

%\author{Cunsheng Ding}
%\ead{cding@ust.hk} 

%\cortext[lcj]{Corresponding author}
%\address{Department of Computer Science and Engineering, 
%The Hong Kong University of Science and Technology,
%Clear Water Bay, Kowloon, Hong Kong, China}

%\tableofcontents

\begin{abstract} 
In this paper, the 3-rank of the incidence matrices of 2-designs 
supported by the minimum weight codewords in a family of ternary linear codes
considered in [C. Ding, C. Li, Infinite families of 2-designs and 3-designs from linear codes,   
Discrete Mathematics 340(10) (2017) 2415--2431] are computed.
 A lower bound on the minimum distance of the ternary codes spanned 
by the incidence matrices of these designs
is derived, and it is proved that the codes are subcodes of 
the 4th order generalized Reed-Muller codes.

\end{abstract}


\end{frontmatter}

%\tableofcontents

\section{Introduction}

Let $\cP$ be a set of $v \ge 1$ elements, and let $\cB$ be a set of $k$-subsets of $\cP$, where $k$ is
a positive integer with $1 \leq k \leq v$. Let $t$ be a positive integer with $t \leq k$. The pair
$\bD = (\cP, \cB)$ is called a $t$-$(v, k, \lambda)$ {\em design\index{design}}, or simply {\em $t$-design\index{$t$-design}}, if every $t$-subset of $\cP$ is contained in exactly $\lambda$ elements of
$\cB$. The elements of $\cP$ are called points, and those of $\cB$ are referred to as blocks.
We usually use $b$ to denote the number of blocks in $\cB$.
  A $t$-design is called {\em simple\index{simple}} if $\cB$ does not contain
any repeated blocks.
In this paper, we consider only simple $t$-designs with $v > k > t$.
A $t$-$(v,k,\lambda)$ design is referred to as a 
{\em Steiner system\index{Steiner system}} if $t \geq 2$ and $\lambda=1$, 
and is denoted by $S(t,k, v)$. 

Let $\bD = (\cP, \cB)$ be a $t$-$(v, k, \lambda)$ design with $b \ge 1$ blocks. 
The points of $\cP$ are usually indexed with $p_1,p_2,\cdots,p_v$, and the 
blocks of $\cB$ 
are normally denoted by $B_1, B_2, \cdots, B_b$.
 The {\em incidence matrix\index{incidence matrix}} $M_\bD=(m_{ij})$ of $\bD$ 
is a $b \times v$ matrix where $m_{ij}=1$ if  $p_j$ is on $B_i$ 
and $m_{ij}=0$ otherwise.
The $p$-rank of a design $\bD$ is defined as the rank of its incidence matrix
over a finite field of characteristic $p$.
 The binary matrix $M_{\bD}$ can be viewed as a matrix over $\gf(q)$ for any 
prime power $q$, and its row vectors span a linear code of length $v$ 
over $\gf(q)$, which is 
denoted by $\C_q(\bD)$ and called the \emph{code} of $\bD$ over $\gf(q)$. 
It is known that the $p$-rank of a $t$-$(v,k,\lambda)$ design
with $t\ge 2$ can be smaller than $v-1$ (hence, the dimension of 
$\C_q(\bD)$, where $q=p^t$, can be smaller than $v-1$),
only if $p$ divides $\lambda_1 - \lambda_2$, where $\lambda_i$
denotes the number of blocks of $\bD$ that contain $i$ points
($i=1, 2$) (cf.  \cite{Hamada73}, \cite[Theorem 1.86, page 686]{Tonchevhb}.)
%It is clear that the code $\C_q(\bD)$ depends on the labelling of the points 
%of $\bD$, but is unique up to column permutations.     

We assume that the reader is familiar with the basics of linear codes and cyclic 
codes.
Throughout this paper, we denote the dual code of $\C$ by $\C^\perp$, 
and the extended code of $\C$ by $\overline{\C}$.
 
%and proceed to introduce the support designs of linear codes directly. 
Let $\C$ be a $[v, \kappa, d]$ linear code over $\gf(q)$. Let $A_i:=A_i(\C)$, which denotes the
number of codewords with Hamming weight $i$ in $\C$, where $0 \leq i \leq v$. The sequence 
$(A_0, A_1, \cdots, A_{v})$ is
called the \textit{weight distribution} of $\C$, and $\sum_{i=0}^v A_iz^i$ is referred to as
the \textit{weight enumerator} of $\C$. For each $k$ with $A_k \neq 0$,  let $\cB_k$ denote
the set of the supports of all codewords with Hamming weight $k$ in $\C$, where the coordinates of a codeword
are indexed by $(p_1, \ldots, p_v)$. Let $\cP=\{p_1, \ldots, p_v\}$.  The pair $(\cP, \cB_k)$
may be a $t$-$(v, k, \lambda)$ design for some positive integer $\lambda$, which is called a 
\emph{support design} of the code, and is denoted by $\bD_k(\C)$. In such a case, we say that the code $\C$ holds a $t$-$(v, k, \lambda)$ 
design\footnote{More generally, a code $\C$ holds (or supports) a 
$t$-$v,k,\lambda)$ design 
$\bD$ if every block of $\bD$ is the support of some codeword of
$\C$ \cite{Tonchev}.}.

%Throughout this paper, we denote the dual code of $\C$ by $\C^\perp$, i
%and the extended code of $\C$ by $\overline{\C}$.   

While most linear codes over finite fields do not hold $t$-designs, some linear codes do hold 
$t$-designs for $t \geq 1$. Studying the linear codes of $t$-designs has been a topic of research for a long time 
\cite{AK92,AK98,AM69,CH92,Dingbk15,HP03,KP95,Tonc93,Tonchev,Tonchevhb}. The objective of this paper 
is to study the linear codes of a family of $2$-designs held in a class of ternary linear codes.  It will 
be shown that these codes are subcodes of the fourth-order generalized Reed-Muller codes 
and support new $2$-designs.

\section{Auxiliary results} 

In this section, we present some auxiliary results that will be needed in later sections. 

\subsection{Designs from linear codes via the Assmus-Mattson Theorem}

The following theorem, proved by Assumus and Mattson \cite{AM69},
 shows that the pair $(\cP, \cB_k)$ defined by 
a linear code is a $t$-design under certain conditions.

\begin{theorem}[Assmus-Mattson Theorem]\label{thm-designAMtheorem} (\cite{AM69}, \cite[p. 303]{HP03})
Let $\C$ be a $[v,k,d]$ code over $\gf(q)$. Let $d^\perp$ denote the minimum distance of $\C^\perp$. 
Let $w$ be the largest integer satisfying $w \leq v$ and 
$$ 
w-\left\lfloor  \frac{w+q-2}{q-1} \right\rfloor <d. 
$$ 
Define $w^\perp$ analogously using $d^\perp$. Let $(A_i)_{i=0}^v$ and $(A_i^\perp)_{i=0}^v$ denote 
the weight distribution of $\C$ and $\C^\perp$, respectively. Fix a positive integer $t$ with $t<d$, and 
let $s$ be the number of $i$ with $A_i^\perp \neq 0$ for $0 \leq i \leq v-t$. Suppose $s \leq d-t$. Then 
\begin{itemize}
\item the codewords of weight $i$ in $\C$ hold a $t$-design provided $A_i \neq 0$ and $d \leq i \leq w$, and 
\item the codewords of weight $i$ in $\C^\perp$ hold a $t$-design provided $A_i^\perp \neq 0$ and 
         $d^\perp \leq i \leq \min\{v-t, w^\perp\}$. 
\end{itemize}
\end{theorem}

The Assmus-Mattson Theorem is a very useful tool for constructing $t$-designs from linear codes, 
and has been successfully employed to construct infinitely many $2$-designs and $3$-designs \cite{Dingbk18}.

\subsection{Designs from linear codes via the automorphism group}

%In this section, we introduce the automorphism approach to obtaining $t$-designs from linear codes. 
%To this end, we have to define the automorphism group of linear codes. 
%We will also present some basic results about this approach.  

The set of coordinate permutations that map a code $\C$ to itself forms a group, which is referred to as 
the \emph{permutation automorphism group\index{permutation automorphism group of codes}} of $\C$
and denoted by $\PAut(\C)$. If $\C$ is a code of length $n$, then $\PAut(\C)$ is a subgroup of the 
\emph{symmetric group\index{symmetric group}} $\Sym_n$.

A \emph{monomial matrix\index{monomial matrix}} over $\gf(q)$ is a square matrix having exactly one 
nonzero element of $\gf(q)$  in each row and column. A monomial matrix $M$ can be written either in 
the form $DP$ or the form $PD_1$, where $D$ and $D_1$ are diagonal matrices and $P$ is a permutation 
matrix. 

The set of monomial matrices that map $\C$ to itself forms the group $\MAut(\C)$,  which is called the 
\emph{monomial automorphism group\index{monomial automorphism group}} of $\C$. Clearly, we have 
$$
\PAut(\C) \subseteq \MAut(\C).
$$

The \textit{automorphism group}\index{automorphism group} of $\C$, denoted by $\GAut(\C)$, is the set 
of maps of the form $M\gamma$, 
where $M$ is a monomial matrix and $\gamma$ is a field automorphism, that map $\C$ to itself. In the binary 
case, $\PAut(\C)$,  $\MAut(\C)$ and $\GAut(\C)$ are the same. If $q$ is a prime, $\MAut(\C)$ and 
$\GAut(\C)$ are identical. In general, we have 
$$ 
\PAut(\C) \subseteq \MAut(\C) \subseteq \GAut(\C). 
$$

By definition, every element in $\GAut(\C)$ is of the form $DP\gamma$, where $D$ is a diagonal matrix, 
$P$ is a permutation matrix, and $\gamma$ is an automorphism of $\gf(q)$.   
The automorphism group $\GAut(\C)$ is said to be $t$-transitive if for every pair of $t$-element ordered 
sets of coordinates, there is an element $DP\gamma$ of the automorphism group $\GAut(\C)$ such that its 
permutation part $P$ sends the first set to the second set.

The next theorem gives a %another 
sufficient condition for a linear code to hold $t$-designs \cite[p. 308]{HP03}. 
 
\begin{theorem}\label{thm-designCodeAutm}
Let $\C$ be a linear code of length $n$ over $\gf(q)$ where $\GAut(\C)$ is $t$-transitive. Then the codewords of any weight $i \geq t$ of $\C$ hold a $t$-design.
\end{theorem}

\subsection{The generalized Reed-Muller codes}

The codes of a family of $2$-designs held in a class of affine-invariant codes that will be studied in 
this paper are in fact subcodes of the fourth-order generalized Reed-Muller ternary codes. 
Hence, we review these codes  and some of their properties in this section. 

Let $\ell$ be a positive integer with $1 \leq \ell <(q-1)m$. The $\ell$-th order 
\emph{punctured generalized Reed-Muller code}\index{punctured generalized Reed-Muller code} 
$\cR_q(\ell, m)^*$ over $\gf(q)$ is the cyclic code of length $n=q^m-1$ with generator polynomial 
\begin{eqnarray}\label{eqn-generatorpolyPGRMcode}
g(x) = \sum_{\myatop{1 \leq j \leq n-1}{ \wt_q(j) < (q-1)m-\ell}} (x - \alpha^j), 
\end{eqnarray}
where $\alpha$ is a generator of $\gf(q^m)^*$ \cite{AK98}. Since $\wt_q(j)$ is a constant function on 
each $q$-cyclotomic coset modulo $n=q^m-1$, $g(x)$ is a polynomial over $\gf(q)$. 

The parameters of the punctured generalized Reed-Muller code $\cR_q(\ell, m)^*$ are known 
and summarized in the next theorem. 

\begin{theorem}\label{thm-GPRMcode}  \cite{AK98} 
For any $\ell$ with $0 \leq \ell <(q-1)m$, $\cR_q(\ell, m)^*$ is a cyclic code over $\gf(q)$ 
with length $n=q^m-1$, dimension 
$$ 
\kappa=\sum_{i=0}^\ell \sum_{j=0}^{m} (-1)^j \binom{m}{j} \binom{i-jq+m-1}{i-jq} 
$$
and minimum weight $d=(q-\ell_0)q^{m-\ell_1-1}-1$, where $\ell=\ell_1(q-1)+\ell_0$ and 
$0 \leq \ell_0 < q-1$. 
\end{theorem} 

%\begin{theorem}\label{thm-dualPGRMcode}
For $0 \leq \ell < m(q-1)$, the code $(\cR_q(\ell, m)^*)^\perp$ is the cyclic code with generator polynomial 
\begin{eqnarray}\label{eqn-dualPGRMcodegenerator}
g^\perp(x) = \sum_{\myatop{0 \leq j \leq n-1}{ \wt_q(j) < \ell}} (x - \alpha^j), 
\end{eqnarray}
where $\alpha$ is a generator of $\gf(q^m)^*$. In addition, 
$$ 
(\cR_q(\ell, m)^*)^\perp = (\gf(q)\bone)^\perp \cap \cR_q(m(q-1)-1-\ell, m)^*, 
$$
where $\bone$ is the all-one vector in $\gf(q)^n$ and $\gf(q)\bone$ denotes the code over $\gf(q)$ with 
length $n$ generated by $\bone$.  
%\end{theorem}

The parameters of the dual of the punctured generalized Reed-Muller code are summarized as follows \cite[Section 5.4]{AK92}. 
%\begin{corollary} 
For $0 \leq \ell < m(q-1)$, the code $(\cR_q(\ell, m)^*)^\perp$ has length $n=q^m-1$, dimension 
$$ 
\kappa=n-\sum_{i=0}^\ell \sum_{j=0}^{m} (-1)^j \binom{m}{j} \binom{i-jq+m-1}{i-jq},  
$$
and minimum weight 
\begin{eqnarray}\label{eqn-lbmwtdualGPRMcode}
d \geq (q-\ell'_0)q^{m-\ell'_1-1},
\end{eqnarray}
where $m(q-1)-1-\ell=\ell'_1 (q-1)+\ell'_0$ and 
$0 \leq \ell'_0 < q-1$. 
%\end{corollary} 

The generalized Reed-Muller code $\cR_q(\ell, m)$ is defined to be the extended code 
of $\cR_q(\ell, m)^*$, and its parameters are given below \cite{AK98}. 
%\begin{theorem}\label{thm-GRMcodePara}
Let $0 \leq \ell < q(m-1)$. Then the generalized Reed-Muller code $\cR_q(\ell, m)$ has length $n=q^m$, dimension 
$$ 
\kappa=\sum_{i=0}^\ell \sum_{j=0}^{m} (-1)^j \binom{m}{j} \binom{i-jq+m-1}{i-jq},  
$$
and minimum weight 
\begin{eqnarray*}
d = (q-\ell_0)q^{m-\ell_1-1},
\end{eqnarray*}
where $\ell=\ell_1 (q-1)+\ell_0$ and 
$0 \leq \ell_0 < q-1$.  
%\end{theorem} 

\begin{theorem}\label{thm-DGM70}   \cite{AK98} 
Let $0 \leq \ell < q(m-1)$ and $\ell=\ell_1 (q-1)+\ell_0$, where $0 \leq \ell_0 < q-1$.  
The total number $A_{(q-\ell_0)q^{m-\ell_1-1}}$ of minimum weight codewords in $\cR_q(\ell, m)$ 
is given by 
\begin{eqnarray*}
A_{(q-\ell_0)q^{m-\ell_1-1}} =
(q-1) \frac{q^{\ell_1} (q^{m}-1) (q^{m-1}-1) \cdots (q^{\ell_1+1}-1)}{(q^{m-\ell_1}-1)(q^{m-\ell_1-1}-1) \cdots (q-1)}N_{\ell_0}, 
\end{eqnarray*} 
where 
\begin{eqnarray*}
N_{\ell_0} = \left\{ 
\begin{array}{ll}
1   & \mbox{ if } \ell_0=0, \\ 
\binom{q}{\ell_0} \frac{q^{m-\ell_1}-1}{q-1} & \mbox{ if } 0 < \ell_0 < q-1.  
\end{array} 
\right. 
\end{eqnarray*}
\end{theorem}

The generalized Reed-Muller codes $\cR_q(\ell, m)$ can also be defined with a multivariate polynomial approach. 
The reader is referred to \cite[Section 5.4]{AK98} for details. For $\ell < (q-1)m$, it was shown in \cite{AK98} that  
$$ 
\cR_q(\ell, m)^\perp = \cR_q(m(q-1)-1-\ell, m). 
$$

The general affine group $\GA_1(\gf(q))$ is defined by 
$$ 
\GA_1(\gf(q))=\{ax+b: a \in \gf(q)^*, \ b \in \gf(q) \}, 
$$ 
which acts on $\gf(q)$ doubly transitively \cite[Section 1.7]{Dingbk18}.  A linear code $\C$ 
of length $q$ is said to be affine-invariant if $\GA_1(\gf(q))$ fixes $\C$ \cite{Charpin90}. For affine-invariant 
codes we use the elements of $\gf(q)$ to index the coordinates of their codewords. 

Let $\ell$ be a positive integer with $1 \leq \ell <(q-1)m$. Then $\cR_q(\ell, m)$ 
is affine-invariant, and the automorphism group $\Aut(\cR_q(\ell, m))$ is doubly transitive. These are well 
known facts about the generalized Reed-Muller codes $\cR_q(\ell, m)$ \cite{AK98,Dingbk18}.

\section{Codes of designs held in a class of affine-invariant ternary codes} \label{sec-ternarydesigncode}

Let $p$ be an odd prime and $m \geq 2$ be an integer. Define 
\begin{eqnarray}\label{eqn-feb241}
\C(m,p)=\{(\tr_{p^m/p}(ax^2+bx)+h)_{x \in \gf(p^m)}: a, \ b \in \gf(p^m), \ h \in \gf(p)\}. 
\end{eqnarray} 
Clearly, the code $\C(m,p)$ is affine-invariant, and holds $2$-designs for each fixed nonzero weight 
(see \cite[Section 6.2]{Dingbk18} and \cite{DWF19}).  
Let $d$ denote the minimum weight of $\C(m,p)$. 
Let $\bD_d(\C(m,p))$ denote the design formed by the supports of 
the minimum weight codewords in $\C(m,p)$, and let $\C_p(\bD_d(\C(m,p)))$ 
denote the linear code over $\gf(p)$ spanned by the incidence matrix of 
the design $\bD_d(\C(m,p))$. 
An interesting problem is to determine the parameters of the code  
$\C_p(\bD_d(\C(m,p)))$. 
This problem is hard to solve for general odd $p$, but is feasible 
in the case $p=3$.  

Our objective of this paper is 
to compute the dimension of the code  $\C_3(\bD_d(\C(m,3)))$,
or equivalently, to determine the 3-rank of the incidence matrix of 
$\bD_d(\C(m,3))$,
 and to prove a lower bound on the minimum distance of 
the code   $\C_3(\bD_d(\C(m,3)))$. 
We will also prove that $\C_3(\bD_d(\C(m,3)))$ is a subcode of the 
fourth-order generalized Reed-Muller ternary code. 
 
In the rest of this section below, we fix $q=3^m$ and let $\tr(x)$ denote the trace function from $\gf(q)$ to $\gf(3)$. 
The code $\C(m,3)$ has four nonzero weights when $m$ is odd, and six nonzero weights 
when $m$ is even \cite{DWF19}. 
When $m$ is odd, the code $\C(m,3)$ has parameters $[3^m, 2m+1, 2 \times 3^{m-1}-3^{(m-1)/2}]$, 
and the weight distribution of the code $\C(m,3)$ is given in Table \ref{tab-CG328c}  \cite{Dingbk18}. The dual 
code has parameters $[3^m, 3^m-1-2m, 5]$. Hence, the Assmus-Mattson theorem can also be employed to prove 
that the codewords of a fixed weight in $\C(m,3)$ support a $2$-design \cite{DingLi16}. When $m \geq 3$ is odd, the minimum distance $d=2\times 3^{m-1}-3^{(m-1)/2}$, and the 
design  $\bD_d(\C(m,3))$ has parameters $2$-$(3^m, d, d(d-1)/2)$ \cite{DingLi16}. 
We treat only the case that $m \geq 3$ is odd. At the end of this section, we will state the conclusions for even 
$m$, but will skip their proofs.  

\begin{table}[ht]\label{tab:wd}
\begin{center} 
\caption{Weight distribution of some ternary linear codes}\label{tab-CG328c}
\begin{tabular}{|l|l|}
\hline
Weight $w$    & No. of codewords $A_w$  \\ \hline
$0$                                                        & $1$  \\ 
$2\times 3^{m-1}-3^{(m-1)/2}$           & $3^{2m}-3^m$  \\ 
$2\times 3^{m-1}$                            & $(3^m+3)(3^m-1)$ \\ 
$2\times 3^{m-1}+3^{(m-1)/2}$          & $3^{2m}-3^m$ \\ 
$3^m$                                  &  $2$ \\ \hline 
\end{tabular} 
\end{center} 
\end{table}

We first prove the following result. 

\begin{lemma} 
Let $m \geq 2$. Then $\bone \in \C(m,3)^\perp$. 
\end{lemma} 

\begin{proof}
The desired conclusion follows from the definition of $\C(m,3)$ and 
$$ 
\sum_{x \in \gf(q)} x^2=0 \mbox{ and } \sum_{x \in \gf(q)} x = 0. 
$$
\end{proof} 

We will nee the next auxiliary result. 

\begin{lemma}\label{lem-July0a} 
Let $m \geq 4$. For each $(a, h) \in \gf(q)^* \times \gf(3)$ define 
$$ 
H(a, h)=\{x \in \gf(q): \tr(ax^2)+h=0\}. 
$$ 
For each nopnemptry set $S \subset \gf(q)$, define $\Delta(S)=\{s-s': s, \ s \in S\}$. 
Then $H(a, h)=-H(a, h)$ and $\Delta(H(a, h)) = \gf(q)$. 
\end{lemma}  

\begin{proof}
For any $b \in \gf(q)$, it suffices to show there is a pair $(x, y) \in H(a, h) \times H(a,h)$ such that $y-x=b$. 
This conclusion is obvious for $b =0$, as $H(a,h)$ is not empty. Hence, we need to prove the conclusion 
for all $b \neq 0$. 
This is to prove that the system of equations 
\begin{eqnarray}\label{eqn-tding10}
\tr(ax^2)+h=0 \mbox{ and } \tr(a(x+b)^2)+h=0 
\end{eqnarray}
has at least one solution $x \in \gf(q)$ for each nonzero $b$. 

Let $N(a,b,h)$ denote the number of solutions $x \in \gf(q)$ 
of Equation (\ref{eqn-tding10}) for $b \neq 0$. Let $\chi'_1$ and $\chi_1$ denote the canonical characteristic of $\gf(3)$ and $\gf(q)$, respectively. Recall that $a \ne 0$ and $b \ne 0$. We have then 
\begin{eqnarray*}
\lefteqn{3^2N(a,b,h)} & &  \\
    & = &  \sum_{x\in\gf(q)} \sum_{y_1,y_2\in\gf(3)}
               \chi_1'{\{y_1[\tr(ax^2)+h]+y_2[\tr(a(x+b)^2)+h)]\}}\\
        & = &  \sum_{y_1,y_2\in\gf(3)} \sum_{x\in\gf(q)}
               \chi_1'(y_1h+y_2h)\chi_1[(y_1+y_2)ax^2+2aby_2 x+y_2ab^2] \\
        & = & \sum_{x\in\gf(q)} 1   
              + \sum_{y_2\in\gf(3)^*} 
               \sum_{x\in\gf(q)} \chi_1'(0) \chi_1(2aby_2 x+y_2ab^2) \\ 
         & & +  \sum_{y_1\neq -y_2} \sum_{x\in\gf(q)}
               \chi_1'(y_1h+y_2h)\chi_1[(y_1+y_2)ax^2+2aby_2 x+y_2ab^2]  \\
        & = & 3^m+ 
 \sum_{y_1\neq -y_2} \sum_{x\in\gf(q)}
               \chi_1'(y_1h+y_2h)\chi_1[(y_1+y_2)ax^2+2aby_2 x+y_2ab^2]  \\
        & = & 3^m+ 
 \sum_{y_1\neq -y_2}   \chi_1'(y_1h+y_2h)   \sum_{x\in\gf(q)}
               \chi_1[(y_1+y_2)ax^2+2aby_2 x+y_2ab^2]                 
\end{eqnarray*}
It then follows from the Weil bound on exponential sums in \cite[p. 218]{Lidl97} that 
\begin{eqnarray*}
\left| 3^2N(a,b,h) -3^m  \right|
&=&  \left|   \sum_{y_1\neq -y_2}   \chi_1'(y_1h+y_2h)   \sum_{x\in\gf(q)}
               \chi_1[(y_1+y_2)ax^2+2aby_2 x+y_2ab^2]        \right| \\
 & \leq &   \sum_{y_1\neq -y_2}   |\chi_1'(y_1h+y_2h)|  \left|   \sum_{x\in\gf(q)}
               \chi_1[(y_1+y_2)ax^2+2aby_2 x+y_2ab^2]        \right| \\          
  & = &  \sum_{y_1\neq -y_2}    \left|   \sum_{x\in\gf(q)}
               \chi_1[(y_1+y_2)ax^2+2aby_2 x+y_2ab^2]        \right| \\      
  & \leq & 6 \times 3^{m/2}.                            
\end{eqnarray*} 
Consequently, 
$$ 
N(a,b,h) \geq 3^{m-2}-2 \times 3^{(m-2)/2} > 1
$$ 
for $m \geq 4$. This completes the proof. 
\end{proof}

\begin{lemma}\label{lem-July0b} 
Let $m \geq 3$. Define by $Q(q)$ and $N(q)$ the set of nonzero squares and the set of nonsquares in 
$\gf(q)$. Then $\Delta(Q(q)) =  \gf(q)$ and  $\Delta(N(q)) =  \gf(q)$. 
\end{lemma} 

\begin{proof}
It is known that $Q(q)$ and $N(q)$ are $(q, (q-1)/2, (q-3)/4)$ difference sets in $(\gf(q), +)$ for odd  
$m$, and $(q, (q-1)/2, (q-5)/4, (q-1)/2)$ almost difference sets in $(\gf(q), +)$ for all  even $m$ \cite{Dingbk15}.  
The desired conclusions then follow. 
\end{proof}

The proof of the following lemma is easy and omitted.  

\begin{lemma} 
For each $(a, b) \in \gf(q)^* \times \gf(q)$, there is exactly one $h \in \gf(3)$ such that the codeword 
$\bc_{(a,b,h)}:=(\tr(ax^2+bx)+h)_{x \in \gf(q)}$ has minimum Hamming weight $d$.  Hence, the total 
number of minimum weight codewords is $q(q-1)$. 
\end{lemma} 

For each $(a, b) \in \gf(q)^* \times \gf(q)$, let $h(a,b)$ denote the unique element in $\gf(q)$ 
such that the codeword $\bc_{(a,b,h)}:=(\tr(ax^2+bx)+h(a,b))_{x \in \gf(q}$ has the minimum weight. 
We use the elements of $\gf(q)$ to index the coordinates of the code 
$\C_3(\bD_d(\C(m,3)))$. We also use $\gf(q)$ as the point set of the 
design $\bD_d(\C(m,3))$. 

For each $(a, b) \in \gf(q)^* \times \gf(q)$, define a vector 
$$ 
\bg(a, b, h(a,b))=(g(a,b,h(a,b))_x)_{x \in \gf(q)} \in \gf(3)^{q}, 
$$
where 
\begin{eqnarray*}
g(a,b,h(a,b))_x = 
\left\{ 
\begin{array}{ll}
1 & \mbox{ if } \tr(ax^2+bx) +h(a, b) \neq 0, \\
0 & \mbox{ otherwise.} 
\end{array}
\right. 
\end{eqnarray*}
By definition, each vector $\bg(a, b, h(a,b))$ has minimum weight $d$, 
and the code $\C_3(\bD_d(\C(m,3)))$ is the linear subspace spanned by the vectors in 
the following set 
$$ 
\{\bg(a, b, h(a,b)), \ (a, b) \in \gf(q)^* \times \gf(q) \}. 
$$ 

By definition, for each $x \in \gf(q)$ we have 
$$ 
g(a,b,h(a,b))_x = ( \tr(ax^2+bx) +h(a, b) )^2. 
$$ 
This expression will help us analyze the code $\C_3(\bD_d(\C(m,3)))$. 

Recall that $\C(m,3)$ is affine-invariant. Let $m$ be odd and let $h \in \gf(3)$ be the unique element such that the codeword 
$( \tr(x^2) +h)_{x \in \gf(q)}$ has minimum weight in $\C(m,3)$. It is easily seen that the set of all minimum weight 
codewords in $\C(m,3)$ is given by 
$$ 
\{ \pm \left ( \tr((ax+b)^2) +h \right )_{x \in \gf(q)}: (a, b) \in \gf(q)^* \times \gf(q) \}
$$ 
Consequently, we obtain the following lemma. 

\begin{lemma}\label{lem-bas} 
Let $m \geq 3$ be odd. 
 Then the code $\C_3(\bD_d(\C(m,3)))$ is linearly spanned by the vectors in following set: 
 \begin{eqnarray}
 \{ ( \tr((ax+b)^2) +h)^2_{x \in \gf(q)}: (a, b) \in \gf(q)^* \times \gf(q) \} 
 \end{eqnarray}
\end{lemma} 

In the following, we identify any vector $(f(x))_{x \in \gf(q)} \in \gf(3)^q$ with the 
function $f(x)$ from $\gf(q)$ to $\gf(3)$. This will simplify our discussions below.  We are now ready 
to determine the dimension of the code $\C_3(\bD_d(\C(m,3)))$ and derive a lower bound on the 
minimum distance of the code. 

Let $a \in \gf(q)^*$. Note that 
\begin{eqnarray}\label{eqn-july1} 
\lefteqn{ \left( \tr((ax+b)^2)+h  \right)^2 =  \tr(a^2x^2)^2 +  \tr(abx)^2 + (\tr(b^2)+h)^2 + }  \nonumber \\ 
&   \tr(a^2x^2)\tr(abx) -(\tr(b^2)+h) \tr(a^2x^2) +  (\tr(b^2)+h) \tr(abx) \in \C_3(\bD_d(\C(m,3))). 
\end{eqnarray} 
Replacing $b$ with $-b$ in (\ref{eqn-july1}), we obtain 
\begin{eqnarray}\label{eqn-july2} 
\lefteqn{ \left( \tr((ax-b)^2)+h  \right)^2 =  \tr(a^2x^2)^2 +  \tr(abx)^2 + (\tr(b^2)+h)^2 - }  \nonumber \\ 
&   \tr(a^2x^2)\tr(abx) -(\tr(b^2)+h) \tr(a^2x^2) -  (\tr(b^2)+h) \tr(abx) \in \C_3(\bD_d(\C(m,3))). 
\end{eqnarray} 
Subtracting  (\ref{eqn-july2}) from (\ref{eqn-july1}) yields 
$$ 
 \tr(a^2x^2)\tr(abx) +  (\tr(b^2)+h) \tr(abx) \in \C_3(\bD_d(\C(m,3))), 
$$ 
which is the same as 
\begin{eqnarray}\label{eqn-july3} 
 \tr(a^2x^2)\tr(cx) +  (\tr((c/a)^2)+h) \tr(cx) \in \C_3(\bD_d(\C(m,3)))  
\end{eqnarray} 
for all $(a, c) \in \gf(q)^* \times \gf(q)$. 

Adding  (\ref{eqn-july2}) and (\ref{eqn-july1}) gives  
\begin{eqnarray}\label{eqn-cont-1} 
 \tr(a^2x^2)^2 +  \tr(abx)^2 + (\tr(b^2)+h)^2 - (\tr(b^2)+h) \tr(a^2x^2)  \in \C_3(\bD_d(\C(m,3))), 
\end{eqnarray} 
which is the same as 
\begin{eqnarray}\label{eqn-july4} 
 \tr(a^2x^2)^2 +  \tr(cx)^2 + (\tr((c/a)^2)+h)^2 - (\tr((c/a)^2)+h) \tr(a^2x^2)  \in \C_3(\bD_d(\C(m,3)))  
\end{eqnarray} 
for all $(a, c) \in \gf(q)^* \times \gf(q)$.

\begin{lemma}\label{lem-aver-Tr}
Let $m \geq 3$ be an odd integer  and $q=3^m$.  Let $b, x\in \mathrm{GF}(q)$. Then the following hold. 
\begin{enumerate}
\item $\sum_{a\in \mathrm{GF}(q)^*} \mathrm{Tr}(a^2x^2)^2=0$.
\item $\sum_{a\in \mathrm{GF}(q)^*} \mathrm{Tr}(abx)^2=0$.
\item $\sum_{a\in \mathrm{GF}(q)^*} \mathrm{Tr}(a^2x^2)=0$.
\end{enumerate}
\end{lemma} 
\begin{proof}
If $x=0$, the conclusions are obvious. Next, let $x\neq 0$.
Then,
\begin{align*}
\sum_{a\in \mathrm{GF}(q)^*} \mathrm{Tr}(a^2x^2)^2
=& \sum_{a\in \mathrm{GF}(q)} \mathrm{Tr}(a^2x^2)^2\\
=&  \sum_{a\in \mathrm{GF}(q)} \mathrm{Tr}(a^2)^2\\
=& | \{a\in \mathrm{GF}(q): \mathrm{Tr}(a^2)\neq 0\} |  \bmod{3} \\
=& \mathrm{wt}\left (  (\mathrm{Tr}(x^2))_{x\in \mathrm{GF}(q)}\right ) \bmod{3} \\
=& 0,
\end{align*}
where the last equality follows from Table \ref{tab:wd}.

If $b=0$,  $\sum_{a\in \mathrm{GF}(q)^*} \mathrm{Tr}(abx)^2=0$. Next, let $b\neq 0$.
Then
\begin{align*}
\sum_{a\in \mathrm{GF}(q)^*} \mathrm{Tr}(abx)^2=&\sum_{a\in \mathrm{GF}(q)} \mathrm{Tr}(abx)^2\\
=&\sum_{a\in \mathrm{GF}(q)} \mathrm{Tr}(a)^2\\
=&  |\{a\in \mathrm{GF}(q): \mathrm{Tr}(a)\neq 0\}| \bmod{3} \\
=&2\cdot 3^{m-1} \bmod{3} \\
=&0.
\end{align*}
Since $\sum_{a\in \mathrm{GF}(q)^*} a^2=0$, we have 
\begin{align*}
\sum_{a\in \mathrm{GF}(q)^*} \mathrm{Tr}(a^2x^2)=\sum_{a\in \mathrm{GF}(q)^*} \mathrm{Tr}(a^2) 
=  \mathrm{Tr}(\sum_{a\in \mathrm{GF}(q)^*}a^2) 
= 0.
\end{align*}
This completes the proof.

\end{proof}

\begin{lemma}\label{lem-july197} 
Let $m \geq 4$. We have $1  \in \C_3(\bD_d(\C(m,3)))$. 
\end{lemma} 

\begin{proof}
By (\ref{eqn-cont-1}), we have
\begin{align*}
\sum_{a\in \gf(q)^*} \left (\tr(a^2x^2)^2 +  \tr(abx)^2 + (\tr(b^2)+h)^2 - (\tr(b^2)+h) \tr(a^2x^2) \right ) \in \C_3(\bD_d(\C(m,3))).
\end{align*}
 It then follows from  Lemma \ref{lem-aver-Tr}
that $(\tr(b^2)+h)^2   \in \C_3(\bD_d(\C(m,3)))$ for all $b \in \gf(q)$. One can find  $b \in  \gf(q)$ such that $\tr(b^2)+h \neq 0$. The desired 
conclusion then follows.  
\end{proof}

\begin{lemma}\label{lem-july191} 
Let $m \geq 5$ be odd. 
For all $(a, c) \in \gf(q) \times \gf(q)$, $\tr(ax^2)\tr(cx) \in \C_3(\bD_d(\C(m,3)))$. 
\end{lemma} 

\begin{proof}
The conclusion is obvious for $a=0$ or $c=0$. We now assume that $a, c \in \gf(q)^*$. Recall the set $H(a^{-2},h)$ defined 
in Lemma \ref{lem-July0a}. By Lemma  \ref{lem-July0a}, there is a basis $\{c_1, c_2, \ldots, c_m\} \subset H(a^{-2},h)$ 
of $\gf(q)$ over $\gf(3)$. It then follows from \eqref{eqn-july3} that $\tr(a^2x^2)\tr(c_ix) \in \C_3(\bD_d(\C(m,3)))$ 
for all $i$. Consequently, $\tr(a^2x^2)\tr(cx) \in \C_3(\bD_d(\C(m,3)))$ for all $c \in \gf(q)$ and all $a \in \gf(q)^*$.   

Recall the set $Q(q)$ defined in Lemma  \ref{lem-July0b}. 
For each fixed $c \in \gf(q)^*$, by Lemma  \ref{lem-July0b} there is a basis $\{a_1, a_2, \ldots, a_m\} \subset Q(q)$ 
of $\gf(q)$ over $\gf(3)$ such that $\tr(a_ix^2)\tr(cx) \in \C_3(\bD_d(\C(m,3)))$. It then follows that 
$\tr(ax^2)\tr(cx) \in \C_3(\bD_d(\C(m,3)))$ for all $a \in \gf(q)$. The desired conclusion then follows. 
\end{proof}

\begin{lemma}\label{lem-july192} 
Let $m \geq 4$. 
For all $c \in \gf(q)$, $\tr(cx) \in \C_3(\bD_d(\C(m,3)))$. 
\end{lemma} 

\begin{proof}
By Lemma \ref{lem-july191}, $\tr(a^2x^2)\tr(cx) \in \C_3(\bD_d(\C(m,3)))$. It then follows from \eqref{eqn-july3}  
that  
$$(\tr((c/a)^2)+h) \tr(cx) \in \C_3(\bD_d(\C(m,3)))$$ 
for all $a \in \gf(q)^*$ and $c \in \gf(q)$. For each $c$ there 
is an $a \in \gf(q)$ such that  $\tr((c/a)^2)+h=1$. Consequently,  $\tr(cx) \in \C_3(\bD_d(\C(m,3)))$. This completes 
the proof.   
\end{proof}

\begin{lemma}\label{lem-july193} 
Let $m \geq 4$. 
For all $b_1, b_2 \in \gf(q)$, $\tr(b_1x) \tr(b_2x)  \in \C_3(\bD_d(\C(m,3)))$. 

\end{lemma} 

\begin{proof}
One can find  $a \in  \gf(q)^*$ such that $\tr(\frac{b_1b_2}{a^2}) = 0$.
Let $c_1=\frac{-b_1-b_2}{a}$ and $c_2=\frac{-b_1+b_2}{a}$. 
It follows from \eqref{eqn-cont-1} and Lemma \ref{lem-july197} that 
\begin{align}\label{sub-c1}
\tr(a^2x^2)^2 +  \tr(ac_1x)^2 - (\tr(c_1^2)+h) \tr(a^2x^2)  \in \C_3(\bD_d(\C(m,3)))
\end{align}
and 
\begin{align}\label{sub-c2}
\tr(a^2x^2)^2 +  \tr(ac_2x)^2 - (\tr(c_2^2)+h) \tr(a^2x^2)  \in \C_3(\bD_d(\C(m,3))).
\end{align}
Subtracting  (\ref{sub-c2}) from (\ref{sub-c1}) yields 
$$(\tr (a(c_1+c_2)x))(\tr (a(c_1-c_2)x))- (\tr(c_1^2-c_2^2))\tr(a^2x^2)\in \C_3(\bD_d(\C(m,3))),$$
which is the same as 
$$\tr (b_1x)\tr (b_2x)- \tr\left (\frac{b_1b_2}{a^2} \right )\tr(a^2x^2)\in \C_3(\bD_d(\C(m,3))).$$
The desired conclusion then follows from $\tr(\frac{b_1b_2}{a^2}) = 0$. 

\end{proof}

As a corollary of Lemma \ref{lem-july193}, we have the following. 

 \begin{lemma}\label{lem-july194} 
Let $m \geq 4$. 
For all $a \in \gf(q)$, $\tr(ax)^2  \in \C_3(\bD_d(\C(m,3)))$. 
\end{lemma}

\begin{lemma}\label{lem-july196} 
Let $m \geq 4$. 
For all $a \in \gf(q)$, $\tr(ax^2)  \in \C_3(\bD_d(\C(m,3)))$. 
\end{lemma}  

\begin{proof}
By Lemma \ref{lem-july194}, $\tr(cx)^2 \in  \C_3(\bD_d(\C(m,3)))$. It then follows from \eqref{eqn-july4} 
that 
\begin{eqnarray}\label{eqn-9july1}
\tr(a^2x^2)^2 +  (\tr((c/a)^2)+h)^2 - (\tr((c/a)^2)+h) \tr(a^2x^2)  \in \C_3(\bD_d(\C(m,3))) 
\end{eqnarray} 
for all $a \in \gf(q)^*$ and $c \in \gf(q)$. Choose $c_1$ and $c_2$ in $\gf(q)$ such that 
$$ 
\tr((c_1/a)^2)+h=2 \mbox{ and }  \tr((c_2/a)^2)+h=1. 
$$ 
Plugging $c_1$ and $c_2$ into (\ref{eqn-9july1}) yields 
\begin{eqnarray*}
&& \tr(a^2x^2)^2 +  (2)^2 - 2 \tr(a^2x^2)  \in \C_3(\bD_d(\C(m,3))), \\
&& \tr(a^2x^2)^2 +  (1)^2 - 1 \tr(a^2x^2)  \in \C_3(\bD_d(\C(m,3))). 
\end{eqnarray*} 
Taking the difference of the two functions above shows that $\tr(a^2x^2)  \in \C_3(\bD_d(\C(m,3))$. 
The desired conclusion then follows from Lemma \ref{lem-July0b}.  

\end{proof}

\begin{lemma}\label{lem-july195} 
Let $m\ge 5$ be an odd integer. For all $b_1, b_2 \in \gf(q)$, $\tr(b_1x^2) \tr(b_2x^2)  \in \C_3(\bD_d(\C(m,3)))$. 
\end{lemma}  

\begin{proof}
By Lemma \ref{lem-july194}, $\tr(cx)^2 \in  \C_3(\bD_d(\C(m,3)))$. 
By Lemma \ref{lem-july196}, $\tr(a^2x^2) \in  \C_3(\bD_d(\C(m,3)))$.
It then follows from \eqref{eqn-july4} and Lemma \ref{lem-july197}
that 
\begin{align*}
\tr(a^2x^2)^2 +  (\tr((c/a)^2)+h)^2   \in \C_3(\bD_d(\C(m,3))) 
\end{align*} 
for all $a \in \gf(q)^*$ and $c \in \gf(q)$.  
It then follows Lemma \ref{lem-july197} that $\tr(a^2x^2)^2  \in \C_3(\bD_d(\C(m,3)))$. 
Consequently, 
\begin{align}\label{eq-a^2-a^2}
\tr(a_1^2x^2)^2 - \tr(a_2^2x^2)^2 = \tr((a_1^2+a_2^2)x^2)  \tr((a_1^2-a_2^2)x^2)  \in \C_3(\bD_d(\C(m,3)))  
\end{align}
for all $a_1$ and $a_2$. 
Since $-1$ is a quadratic non-residue, there are $\epsilon_1, \epsilon_2\in \{1, -1\}$ such that $\epsilon_1(b_1+b_2)$ and $\epsilon_2(-b_1+b_2)$
are quadratic residue. Let $a_1, a_2 \in \gf(q)$ such that $\epsilon_1(b_1+b_2)=a_1^2$ and $\epsilon_2(-b_1+b_2)=a_2^2$.
Then, $a_1^2+a_2^2=(\epsilon_1-\epsilon_2)b_1+(\epsilon_1+\epsilon_2)b_2$ and $a_1^2-a_2^2=(\epsilon_1+\epsilon_2)b_1+(\epsilon_1-\epsilon_2)b_2$.
By (\ref{eq-a^2-a^2}), one has
\begin{align*}
\tr(((\epsilon_1-\epsilon_2)b_1+(\epsilon_1+\epsilon_2)b_2)x^2)  \tr(((\epsilon_1+\epsilon_2)b_1+(\epsilon_1-\epsilon_2)b_2)x^2)  \in \C_3(\bD_d(\C(m,3))).
\end{align*}
If $\epsilon_1=\epsilon_2$, $\tr(((\epsilon_1+\epsilon_2)b_2)x^2)  \tr(((\epsilon_1+\epsilon_2)b_1)x^2)= \tr(b_1x^2) \tr(b_2x^2)   \in \C_3(\bD_d(\C(m,3))).$\\
If $\epsilon_1=-\epsilon_2$, $\tr(((\epsilon_1-\epsilon_2)b_2)x^2)  \tr(((\epsilon_1-\epsilon_2)b_1)x^2)= \tr(b_1x^2) \tr(b_2x^2)   \in \C_3(\bD_d(\C(m,3))).$\\
This completes the proof. 
\end{proof}

We are now ready to prove the following theorem. 

\begin{theorem}\label{thm-base1} 
Let $m \geq 5$ be odd. Then $\C_3(\bD_d(\C(m,3)))$ is linearly spanned by the functions in the set  
\begin{eqnarray}\label{set-thm1}
\left\{ \tr(ax^2)\tr(bx^2), \tr(ax^2)\tr(bx), \tr(ax) \tr(bx), \tr(bx), 1: a, b \in \gf(q) \right\}. 
\end{eqnarray} 
\end{theorem}

\begin{proof}
By Lemma \ref{lem-july195}, $ \tr(ax^2)\tr(bx^2) \in \C_3(\bD_d(\C(m,3)))$. 
It follows from Lemma \ref{lem-july191} that all functions $ \tr(ax^2)\tr(bx) \in \C_3(\bD_d(\C(m,3)))$. 
By Lemma \ref{lem-july193}, $\tr(ax)\tr(bx) \in \C_3(\bD_d(\C(m,3)))$. 
By Lemma \ref{lem-july192}, $\tr(bx) \in \C_3(\bD_d(\C(m,3)))$. 
By Lemma \ref{lem-july197}, $1 \in \C_3(\bD_d(\C(m,3)))$.   
Consequently, $\C_3(\bD_d(\C(m,3)))$ contains the linear space spanned by the functions in the set 
in (\ref{set-thm1}). 

By \eqref{eqn-july1}, all the functions $(\tr((ax+b)^2)+h)^2$ can be generated by the functions in the set 
in (\ref{set-thm1}). This completes the proof. 
\end{proof}

For convenience of discussion below, we use $\left\langle  S \right\rangle$ to denote the space over 
$\gf(3)$ linearly spanned by the functions from $\gf(q)$ to $\gf(3)$ in the set $S$. To achieve our 
objective , we have to prove the next lemmas. 

\begin{lemma}\label{lem-july1911} 
Let $m \geq 2$. Then 
\begin{eqnarray}\label{eqn-july1911}
\left\langle \{ \tr(ax)\tr(bx): a, b \in \gf(q)  \} \right\rangle = 
\left\{  \sum_{j=0}^{m-1} \tr(c_j x^{3^j+1}): c_i \in \gf(q)    \right\}.
\end{eqnarray}
\end{lemma} 

\begin{proof}
Let $R$ denote the linear subspace in the right-hand side of \eqref{eqn-july1911}. Let $\{\alpha^{3^i}: 0 \leq i \leq m-1\}$ 
be a normal basis of $\gf(q)$ over $\gf(3)$.  Let $x=\sum_{i=0}^{m-1} x_i \alpha^{3^i}$, where $x_i \in \gf(3)$. Then 
$x^{3^j}=\sum_{h=0}^{m-1} x_{(h-j) \bmod{m}} \alpha^{3^h}$ for all $0 \leq j \leq m-1$. We then deduce that 
\begin{eqnarray*}
x^{3^j+1} = \sum_{\ell =0}^{m-1} \left(  \sum_{\myatop{i+h=\ell \pmod{m}}{0 \leq i, \ h < m}} x_i x_{(h-j) \bmod{m}} \right) \alpha^{3^\ell}
\end{eqnarray*}
for all $0 \leq j \leq m-1$. It then follows that 
$$ 
\alpha^{3^u} x^{3^j+1} =  \sum_{\ell =0}^{m-1} \left(  \sum_{\myatop{i+h=\ell-u \pmod{m}}{0 \leq i, \ h < m}} x_i x_{(h-j) \bmod{m}} \right) \alpha^{3^\ell}
$$ 
and 
$$ 
\tr \left(\alpha^{3^u} x^{3^j+1} \right) = \tr(\alpha)  \sum_{\ell =0}^{m-1} \left(  \sum_{\myatop{i+h=\ell-u \pmod{m}}{0 \leq i, \ h < m}} x_i x_{(h-j) \bmod{m}} \right) 
$$ 
for all $0 \leq j \leq m-1$ and $0 \leq u \leq m-1$. 
Hence, $R$ is linearly spanned by the following  functions 
\begin{eqnarray}\label{eqn-set1}
 \sum_{\ell =0}^{m-1} \left(  \sum_{\myatop{i+h=\ell-u \pmod{m}}{0 \leq i, \ h < m}} x_i x_{(h-j) \bmod{m}} \right), \ 
0 \leq j \leq m-1, \ 0 \leq u \leq m-1.   
\end{eqnarray}

Let $L$ denote the linear subspace in the left-hand side of \eqref{eqn-july1911}. Note that 
$$ 
\tr\left(  \alpha^{3^u} x \right) = \tr(\alpha) \sum_{i=0}^{m-1} x_{(u-i) \bmod{m}}.   
$$ 
We have 
$$ 
\tr\left(  \alpha^{3^u} x \right)\tr\left(  \alpha^{3^v} x \right) = 
\tr(\alpha)^2 \left(  \sum_{i=0}^{m-1} x_{(u-i) \bmod{m}} \right) \left( \sum_{j=0}^{m-1} x_{(v-j) \bmod{m}} \right).  
$$
Hence, $L$ is linearly spanned by the following  functions
\begin{eqnarray}\label{eqn-set2}
 \left(  \sum_{i=0}^{m-1} x_{(u-i) \bmod{m}} \right) \left( \sum_{j=0}^{m-1} x_{(v-j) \bmod{m}} \right), \ 
 0 \leq u \leq m-1, \ 0 \leq v \leq m-1. 
\end{eqnarray}

It can be verified that the set of functions in (\ref{eqn-set2}) is the same as the set of functions in (\ref{eqn-set1}). 
It then follows that $R=L$. 
\end{proof}

\begin{lemma}\label{lem-july1912} 
Let $m \geq 2$. Then 
\begin{eqnarray}\label{eqn-july1912}
\left\langle \{ \tr(ax^2)\tr(bx): a, b \in \gf(q)  \} \right\rangle = 
\left\{  \sum_{j=0}^{m-1} \tr(c_j x^{2\times 3^j+1}): c_i \in \gf(q)    \right\}.
\end{eqnarray}
\end{lemma} 

\begin{proof}
The proof is similar to that of Lemma \ref{lem-july1911} and is omitted. 
\end{proof}

\begin{lemma}\label{lem-july1913} 
Let $m \geq 2$. Then 
\begin{eqnarray}\label{eqn-july1913}
\left\langle \{ \tr(ax^2)\tr(bx^2): a, b \in \gf(q)  \} \right\rangle = 
\left\{  \sum_{j=0}^{m-1} \tr(c_j x^{2\times 3^j+2}): c_i \in \gf(q)    \right\}.
\end{eqnarray}
\end{lemma} 

\begin{proof}
The proof is similar to that of Lemma \ref{lem-july1911} and is omitted. 
\end{proof}

We have now the following trace representation of the code $\C_3(\bD_d(\C(m,3)))$. 

\begin{theorem}\label{thm-main3}
Let $m \geq 4$ be odd. Then $\C_3(\bD_d(\C(m,3)))$ is given by 
\begin{eqnarray*}
\left\{  
\sum_{i=0}^{m-1} \tr(a_ix^{2 \cdot 3^i +2}) + 
\sum_{j=0}^{m-1} \tr(b_jx^{2 \cdot 3^j +1}) + 
\sum_{\ell=0}^{m-1} \tr(c_\ell x^{ 3^\ell +1}) + u: 
\ a_i, b_j, c_\ell \in \gf(q), \ u \in \gf(3)  
\right\} .  
\end{eqnarray*}
\end{theorem}

 \begin{proof}
 Note that 
 $$ 
 \left\{  \tr(b_0x^{2 \cdot 3^0 +1}): b_0  \in \gf(q)\right\}=\left\{ \tr(bx): b \in \gf(q)  \right\}. 
 $$ 
 The desired conclusion then follows from Theorem \ref{thm-base1}, Lemmas \ref{lem-july1911}, 
\ref{lem-july1912}, and \ref{lem-july1913}.   
\end{proof}

We are now ready to prove the following main result of this section. 

\begin{theorem}\label{thm-mainlast} 
For each odd $m \geq 5$, we have 
$$ 
\dim(\C_3(\bD_d(\C(m,3))))=2m^2+1  
$$ 
and the minimum distance $d(\C_3(\bD_d(\C(m,3))))$ of the code $\C_3(\bD_d(\C(m,3)))$ is lower 
bounded by $3^{m-2}$. 
\end{theorem} 

\begin{proof}
Put $n=3^m-1$. Define 
\begin{eqnarray*}
A &=& \{-(2\times 3^i+2): 0 \leq i \leq m-1\} \subset \Z_n,  \\
A_1 &=& \{-(2\times 3^i+2): 0 \leq i \leq (m-1)/2\} \subset \Z_n,  \\
B &=& \{-(2\times 3^i+1): 0 \leq i \leq m-1\} \subset \Z_n,  \\
C &=& \{-( 3^i+1): 0 \leq i \leq m-1\} \subset \Z_n,  \\
C_1 &=& \{-(3^i+1): 0 \leq i \leq (m-1)/2\} \subset \Z_n. 
\end{eqnarray*}  
Put $J=A\cup B \cup C$. 
Let $\beta$ be a primitive element of $\gf(q)$. Denote by $\m_i(x)$ the minimal polynomial 
of $\beta^i$ over $\gf(3)$. Define 
$$ 
H(x)=\lcm(\m_i(x): i \in J\},  
$$ 
where $\lcm$ denotes the least common multiple of a set of polynomials. 
Let $\C_3^m$ denote the cyclic code over $\gf(3)$ of length $n$ with parity-check polynomial 
$H(x)$. By the Delsarte Theorem \cite{Delsarte75} and Theorem \ref{thm-main3}, $\C_3(\bD_d(\C(m,3)))$ is 
permutation-equivalent to the 
augmented code of the extended code of  $\C_3^m$. As a result, 
$$ 
\dim(\C_3(\bD_d(\C(m,3))))=\dim(\C_3^m)+1. 
$$
Below we compute the dimension of $\C_3^m$. 

For each $i \in \Z_n$, the cyclotomic coset modulo $n$ containing $i$ is defined by 
$$ 
S_i=\{i 3^j \bmod n: 0 \leq j \leq m-1\} \subset \Z_n. 
$$  
The dimension of $\C_3^m$ is given as 
$$ 
\dim(\C_3^m) =\left|  \bigcup_{i \in A \cup B \cup C} S_i \right|
$$
The following statements can be verified: 
\begin{itemize}
\item $|S_i|=m$ for all $i \in J$. 
\item $S_i \cap S_j=\emptyset$ for all pairs of distinct $i$ and $j$ in $B$.  
\item $S_i \cap S_j=\emptyset$ for all pairs of distinct $i$ and $j$ in $A_1$, 
          and every element in $A \setminus A_1$ is contained in some $S_i$ for $i \in A_1$. 
\item $S_i \cap S_j=\emptyset$ for all pairs of distinct $i$ and $j$ in $C_1$, 
          and every element in $C \setminus C_1$ is contained in some $S_i$ for $i \in C_1$.    
\item $S_i \cap S_j=\emptyset$ for all $i \in B$ and $j$ in $A_1$. 
\item $S_i \cap S_j=\emptyset$ for all $i \in B$ and $j$ in $C_1$. 
\item $S_i \cap S_j=\emptyset$ for all $i \in A_1$ and $j$ in $C_1$, except $(i,j)=(0,1)$, in which case 
          the two cosets are the same.                   
\end{itemize} 
We then deduce that 
$$ 
\dim(\C_3^m) =(|A_1|+|C_1|-1+|B|)m=\left( \frac{m+1}{2} + \frac{m+1}{2} -1 +m \right)m=2m^2.  
$$
The desired conclusion on the dimension of $\C_3(\bD_d(\C(m,3)))$ then follows. 

By Theorem \ref{thm-base1} the code $\C_3(\bD_d(\C(m,3)))$ is a subcode 
of the fourth-order generalized Reed-Muller code $\cR_3(4, m)$. 
By Theorem \ref{thm-DGM70}, the fourth-order generalised Reed-Muller code 
$\cR_3(4, m)$ has dimension 
$$ 
\kappa=\sum_{i=0}^4 \sum_{j=0}^{m} (-1)^j \binom{m}{j} \binom{i-3j+m-1}{i-3j} 
= \binom{m+3}{4}+\binom{m+2}{3} - \frac{(m-1)m}{2}+1 
> 2m^2+1   
$$ 
and minimum distance $3^{m-2}$. The lower bound on the code  $\C_3(\bD_d(\C(m,3)))$ 
then follows. 
\end{proof}

The lower bound $3^{m-2}$ on the minimum distance of $\C_3(\bD_d(\C(m,3)))$ is not very 
tight. It would be nice if the following open problem can be settled. 

\begin{open} 
Determine the minimum distance of the ternary code  $\C_3(\bD_d(\C(m,3)))$. 
\end{open} 

We remark that the conclusions of Theorem \ref{thm-mainlast} are also true for even $m \geq 6$. 
The proofs of the lemmas and theorems for odd $m$ can be modified slightly to prove the 
conclusions of Theorem \ref{thm-mainlast} for even $m$.  The details are left to the reader. 
For $m \in \{2,3,4\}$, the parameters of  $\C_3(\bD_d(\C(m,3)))$ are given in the following example.  
The example  indicates that the code $\C_3(\bD_d(\C(m,3)))$ has good parameters. 
The dimensions of these codes agree with the formula $2m^2+1$. 

\begin{example}\label{ex25}
Let $d$ denote the minimum weight of $\C(m,3)$. The parameters of 
the code $\C(m,3)$ and $\C_3(\bD_d(\C(m,3)))$ for  $m=2, 3, 4$ are listed below: 
\begin{eqnarray}
\begin{array}{rrr}
m & \C(m,3)    &  \C_3(\bD_d(\C(m,3))) \\
2 & [9,5,4]    & [9,9,1] \\
3 & [27,7,15]  & [27, 19, 6] \\
4 & [81,9,48]  & [81, 33, 21] 
\end{array}
\end{eqnarray}  
The ternary code $\C(3,3)$ is distance-optimal  \cite{Grassl} and has weight distribution 
$$ 
1+702z^{15} +780z^{18}+702z^{21}+2z^{27}. 
$$ 
The ternary code $ \C_3(\bD_d(\C(3,3)))$ is also distance-optimal   \cite{Grassl} and has weight distribution 
\begin{eqnarray*}
&& 1+ 5148z^{6} + 14742z^7+ 84240z^8 + 370500 z^9 + 1314144z^{10} +  
4081428z^{11} + \\ 
&& 10838880z^{12} +  25050870z^{13} +  49975380z^{14} +  87147918z^{15} +  129957048z^{16} + \\
&& 
168370488z^{17} + 187697640 z^{18} + 177251490 z^{19} + 141674832 z^{20} +   94909698 z^{21} + \\
&& 51504336 z^{22}+ 22428900 z^{23} + 
7492680 z^{24} + 1796418 z^{25} +   
273780 z^{26} +  20906 z^{27}.  
\end{eqnarray*}
The weight distributions of $\C(3,3)$ and 
$ \C_3(\bD_d(\C(3,3)))$ demonstrate a big difference between the two codes. 
\end{example}

The following problem would be challenging.  

\begin{open} 
Determine the parameters of the code $\C_p(\bD_d(\C(m,p)))$ for odd $p>3$. 
\end{open} 

We point out that the code $\C_3(\bD_d(\C(m,3)))$ is affine-invariant, thus it holds $2$-designs. 
Therefore, the following open problem would be interesting. 

\begin{open} 
Determine the parameters of the $2$-designs held in the code $\C_3(\bD_d(\C(m,3)))$. 
\end{open}

\section{Summary and concluding remarks} 

The contribution of this paper is the study of the ternary codes 
 $\C_3(\bD_d(\C(m,3)))$ carried out in Section 
 \ref{sec-ternarydesigncode}, where the dimensions of the codes were determined,
 and a lower bound on the 
 minimum distance of the codes was proved. We also proved that the 
codes $\C_3(\bD_d(\C(m,3)))$ are subcodes 
 of the fourth-order generalized Reed-Muller ternary codes.
This shows that the code $\C_3(\bD_d(\C(m,3)))$ is 
 much more complicated than the original code $\C(m,3)$, which is defined by quadratic functions $\tr(ax^2+bx)+c$. 
 The difference between the two codes $\C(m,3)$ and $\C_3(\bD_d(\C(m,3)))$ is also seen in their dimensions.        
The codes $\C_3(\bD_d(\C(m,3)))$ are affine-invariant,
hence the codeworfds of any nonzero weight support $2$-designs.

Several open problems were presented in this paper. 
The reader is cordially invited to settle them. 
The $p$-rank of $t$-designs, i.e., the dimension of the corresponding codes, 
can be used to classify $t$-designs of certain type. 
For example, the 2-rank and 3-rank of Steiner triple and quadruple
 systems were intensively studied and 
employed for counting and classifying Steiner triple and quadruple systems
 \cite{JMTW}, 
 \cite{JT},
 \cite{Osuna},
 \cite{Tsts}, \cite{Tsqs},
\cite{Z16}, \cite{ZZ12}, \cite{ZZ13}, \cite{ZZ13a}.

%\section*{Acknowledgments} 

%C. Ding's research was supported by the Hong Kong Research Grants Council,
%Proj. No. 16300418. C. Tang was supported by National Natural Science Foundation of China (Grant No.
%11871058) and China West Normal University (14E013, CXTD2014-4 and the Meritocracy Research
%Funds)

\end{document}